\documentclass[letterpaper, 10 pt, conference]{ieeeconf}  

\usepackage{xcolor}
\IEEEoverridecommandlockouts                              
\let\labelindent\relax

\usepackage{amsthm,amsmath,amsfonts,verbatim, mathtools,enumerate}
\usepackage{tikz}
\usepackage{textcomp}
\usepackage{hyperref}
\usepackage{lipsum}
\usepackage{support-caption}
\usepackage{subcaption}
\usepackage{algpseudocode}
\usepackage{cite}
\allowdisplaybreaks[2]

\newcommand{\reals}{{\mathbb{R}}}

\newcommand{\range}{{\mathrm{im}}}

\newcommand{\diag}{\mathrm{blkdiag}}

\newcommand{\norm}[1]{\left\lVert#1\right\rVert}
\newcommand{\mnorm}[1]{{\left\vert\kern-0.25ex\left\vert\kern-0.25ex\left\vert #1 
    \right\vert\kern-0.25ex\right\vert\kern-0.25ex\right\vert}}

\newcommand{\mc}{\mathcal}

\newcommand{\mb}{\mathbb}

\newcommand{\A}{\mathcal{A}}
\newcommand{\half}{\frac{1}{2}}

\newtheorem{definition}{Definition} 
\newtheorem{theorem}{Theorem}

\newtheorem{corollary}{Corollary}
\newtheorem{remark}{Remark}
\newtheorem{proposition}{Proposition}

\algnewcommand\algorithmicforeach{\textbf{for each}}
\algdef{S}[FOR]{ForEach}[1]{\algorithmicforeach\ #1\ \algorithmicdo}

\newcommand{\ndim}{n}
\newcommand{\X}{\reals}
\newtheorem{example}{Example}

\begin{document}
\title{Disturbance Decoupling for Gradient-based Multi-Agent Learning with Quadratic Costs}
\author{Sarah H. Q. Li$^{1}$, Lillian Ratliff$^{2}$, Beh\c cet\ A\c c\i kme\c se$^{1}$
\thanks{*This research is partly funded by the following grants:  NSF CNS-1736582 and ONR N00014-17-1-2623.}
\thanks{$^{1}$William E. Boeing  Department of Aeronautics and Astronautics, University of Washington, Seattle. 
        email: {\tt $\{$sarahli,behcet$\}$@uw.edu}
}%
\thanks{$^{2}$Department of Electrical and Computer Engineering, University of Washington, Seattle.
        email: {\tt ratliffl@uw.edu }}%
}

\maketitle

\thispagestyle{empty} 
\begin{abstract}
Motivated by applications of multi-agent learning in noisy environments, this paper studies the robustness of gradient-based learning dynamics with respect to disturbances. While disturbances injected along a coordinate corresponding to any individual player's actions can always affect the overall learning dynamics, a subset of players can be \emph{disturbance decoupled}---i.e., such players' actions are completely unaffected by the injected disturbance. We provide necessary and sufficient conditions to guarantee this property for games with quadratic cost functions, which encompass quadratic one-shot continuous games, finite-horizon linear quadratic (LQ) dynamic games, and bilinear games. Specifically, disturbance decoupling is characterized by both algebraic and graph-theoretic conditions on the learning dynamics, the latter is obtained by constructing a \emph{game graph} based on gradients of players' costs. For LQ games, we show that disturbance decoupling imposes constraints on the controllable and unobservable subspaces of players. For two player bilinear games, we show that disturbance decoupling within a player's action coordinates imposes constraints on the payoff matrices.
Illustrative numerical examples are provided.
\end{abstract}

\section{Introduction}\label{sec:intro}
As the application of learning in multi-agent settings gains traction, game theory has emerged as an informative abstraction for understanding the coupling between algorithms employed by individual players (see, e.g.,~\cite{fudenberg1998theory,mazumdar2018convergence,chasnov:2019aa}).
Due to scalability, a commonly employed class of algorithms in both games and modern machine learning approaches to multi-agent learning is \emph{gradient-based learning}, in which players update their individual actions using the {gradient} of their objective with respect to their action. In the gradient-based learning paradigm, continuous \emph{quadratic games} stand out as a benchmark due to their simplicity and ability to exemplify state-of-the-art multi-agent learning methods such as policy gradient and alternating gradient-descent-ascent~\cite{mazumdar2019policy}.  

Despite the resurgence of interest in learning in games, a gap  exists between algorithmic performance in simulation and physical application in part due to disturbances in measurements~\cite{shalev2017failures}.
Robustness  to environmental noise has been analyzed in a wide variety learning paradigms~\cite{li2019robust,bottou2010large}. 
Most analysis focuses on independent and identically distributed stochastic noise drawn from a stationary distribution. 

In contrast, we study \emph{adversarial disturbance} without any assumptions on its dynamics or bounds on its magnitude. 
Though some work exists on the effects of bounded adversarial disturbance in multi-agent learning~\cite{jiao2016multi}, 
there is limited understanding of how gradient disturbance propagates through the network structure as determined by the coupling of the players' objectives. 
Does gradient-based learning fundamentally contribute to or reduce the propagation of disturbance through player actions? 
Our analysis aims to answer this question for gradient-based multi-agent learning dynamics. The insights we gain provide desiderata to support algorithm synthesis and incentive design, and will lead to improved robustness of multi-agent learning dynamics.

\textbf{Contributions.}  
The main contribution is providing a novel graph-theoretical perspective for analyzing disturbance decoupling in multi-agent learning settings. For quadratic games, we obtain a necessary and sufficient condition, which can be verified in polynomial time, that ensures complete decoupling between the corrupted gradient of one player and the learned actions of another player, stated in terms of algebraic and graph-theoretic conditions. The latter perspective leads to greater insight on the types of cost coupling structures that enjoy disturbance decoupling, and hence, provides a framework for designing agent interactions, e.g., via incentive design or algorithm synthesis. Applied to LQ games, a benchmark for multi-agent policy gradient algorithms, we show that disturbance decoupling enforces necessary constraints on the controllable subspace in relation to the unobservable subspace of individual players. Applied to bilinear games, we show that disturbance decoupling enforces necessary constraints on the players' payoff matrices.

\section{Related Work}\label{sec:litReview}
We study gradient-based learning for $N$--player quadratic games with continuous cost functions and action sets. Convergence guarantees for gradient-based learning are studied from numerous perspectives including game theory~\cite{fudenberg1998theory,ratliff:2016aa,chasnov:2019aa}, control~\cite{shamma2005dynamic}, and machine learning~\cite{zhou:2017aa,mazumdar2018convergence}.

Convergence guarantees for gradient-based learning dynamics under stochastic noise are studied in~\cite{chasnov:2019aa,mazumdar2018convergence, zhou:2017aa}. Despite being an important property to understand for adversarial disturbance, how non-stochastic noise propagates through the player network has no guarantees.

Our analysis draws on geometric control~\cite{trentelman2012control,wonham1974linear, dion2003generic}. In~\cite{trentelman2012control}, algebraic conditions for disturbance decoupling within a single dynamical system is given. In~\cite{dion2003generic}, disturbance decoupling for a single structured dynamical system is studied with frequency-based techniques. In this paper, we provide both the algebraic and graph-theoretic conditions for disturbance decoupling of coupled dynamical systems in gradient-based multi-agent learning.

\section{Continuous games and the game graph model}\label{sec:gameSetUp}
Let $[N] = \{1, 2, \ldots, N\}$ denote the index set where $N\in \mb{N}$. 
For a function $f\in C^r(\X^{n},\reals)$
with $r\geq 2$, 
$D_if=\partial f/\partial x_i$ is the partial derivative with respect to $x_i$. 

Consider an $N$-player continuous game $(f_1,\ldots, f_N)$ where for each $i\in [N]$,  $f_i\in C^r(\reals^{n}, \mb{R})$ with $r\geq 2$ is player $i$'s cost function and $\reals^{n}=\reals^{n_1} \times \ldots \times \reals^{n_N}$ is the joint action space, with $\reals^{n_i}$ denoting player $i$'s action space and $n = \sum_{i = 1}^N n_i$. 
Each player's goal is to select an action $x_i \in \X^{n_i}$ to minimize its cost  $f_i: \X^{n} \rightarrow \reals$ given the actions of all other players.
That is, player $i$ seeks to solve the following optimization problem:
\begin{equation}\label{eqn:individualoptimization}
\min_{x_i\in \X^{n_i}} \ f_i(\underbrace{x_1, \ldots, x_i,\ldots, x_N}_{\displaystyle := x}).
\end{equation}
One of the most common characterizations of the outcome of a continuous game is a Nash equilibrium.
\begin{definition}[Nash equilibrium] 
For an $N$--player continuous game $(f_1, \ldots, f_N)$, a joint action $x^\star = (x^\star_1, \ldots, x^\star_N)\in \mathbb{R}^n$  is a Nash equilibrium if  for each $i \in [N]$,
\begin{equation*}
 f_i(x^\star) \leq  f_i(x^\star_1, \ldots,x_{i-1}^\star, x_i, x_{i+1}^\star,\ldots, x^\star_N), \ \forall  \ x_i \in \mathbb{R}^{n_i}.
\end{equation*}
\end{definition}
\subsection{Gradient-based learning}
We consider a class of simultaneous play, gradient-based multi-agent learning techniques such that at iteration $k$, player $i$ receives $h_i(x^k)$ from an oracle to update its action as follows:
\begin{equation}\label{eqn:gradientUpdate}
    \textstyle x_i^{k+1}= x_i^{k} - \gamma_i h_i(x_1^k, \ldots, x_N^k),
\end{equation}
where $\gamma_i > 0$ is player $i$'s step size,
\begin{equation}\label{eqn:noiseModel}
  \textstyle h_i(x^k) = D_if_i(x^k) + d_i^k 
\end{equation}
is player $i$'s gradient evaluated at the current joint action $x^k$ and affected by a player-specific, arbitrary additive disturbance  $d_i^k\in \reals^{n_i}$. In the setting we analyze, $d^k_i$ can modify $x^k_i$ to any other action within $\mathbb{R}^{n_i}$.

Under reasonable assumptions on  step sizes---e.g., relative to the spectral radius of the Jacobian of $h_i$ in a neighborhood of a critical point---it is known that the undisturbed dynamics converge~\cite{mazumdar2018convergence, chasnov:2019aa}. 
While such a guarantee cannot be given for arbitrary disturbances as considered in this paper, we provide conditions under which a subset of players still equilibriates and follows the undisturbed dynamics.
\subsection{Quadratic games}
For an $N$--player continuous game $(f_1, \ldots, f_N)$, behavior of gradient-based learning around a local Nash equilibrium can be approximated by linearizing the learning dynamics, where the \emph{linearization} corresponds to a quadratic game.

\begin{definition}[Quadratic game]\label{def:quadGame}
For each $i \in [N]$, $f_i:\mb{R}^{n}\to\mb{R}$ is defined by 
\begin{equation}\label{eqn:quadCost}
\textstyle f_i(x) = \half x_i^\top P_ix_i + x_i^\top(\textstyle {\sum}_{j\neq i} P_{ij}x_j + r_i).
\end{equation}
\end{definition}
Quadratic games encompass potential games~\cite{monderer1996potential} with $P_{ij} = P_{ji}^\top$, and zero sum games~\cite{gillies1959solutions} with  $P_{ij} = -P_{ji}^\top$. We give further examples of quadratic games in Section~\ref{sec:quadGame_subClasses}. 
\subsection{Game graph}\label{sec:gameGraph}
To highlight how an individual player's action updates depend on others' actions, we associate a directed graph to the gradient-based learning dynamics defined in \eqref{eqn:gradientUpdate}. 

We consider a directed graph $([N], \mc{E})$, where $[N]$ is the index set for the nodes in the graph, and $\mc{E}$ is the set of edges. Each node $i \in [N]$ is associated with action $x_i$ of the $i^{th}$ player. A directed edge $(j, i)$ points from $j$ to $i$ and has weight matrix $W_{ij} \in \reals^{n_i\times n_j}$, such that $(j, i) \in \mc{E}$ if $W_{ij} \neq 0$ element-wise.
For each node $i$, we assume the self loop edge $(i, i)$ always exists and has weight $W_{ii} \in \reals^{n_i \times n_i}$. The composite matrix $W \in \reals^{n \times n}$ with entries $W_{ij}$ is the adjacency matrix of the \emph{game graph}.  

On a game graph, we define a path $p = (i, v_1, \ldots, v_{k-1}, j)$ as a sequence of nodes connected by edges. The set of paths $ \mc{P}^k_{ij}$ includes all paths starting at $i$ and ending at $j$, traversing $k+1$ nodes in total. For a path $p = (i, v_1, \ldots, v_{k-1}, j)$, we define its \emph{path weight} as the product of consecutive edges on the path, given by $W_{j, v_{k-1}}\ldots W_{v_1, i} =  \prod^{k-1}_{l = 0} W_{v_{l+1}, v_{l}}$.

In the absence of disturbances $d_i$, the update in~\eqref{eqn:gradientUpdate} for a quadratic game reduces to
\begin{equation}\label{eqn:gameGraph_dynamics}
 \textstyle x^{k+1} = W x^k - \Gamma \bar{r} ,  
\end{equation}
where $\bar{r} = \begin{bmatrix}r_1^\top& \ldots& r_N^\top\end{bmatrix}^\top$, $W_{ii} = I_{n_i} - \gamma_i P_i$, $W_{ij} = -\gamma_iP_{ij}$,  and $\Gamma=\diag(\gamma_1 I_{n_1}, \ldots, \gamma_N I_{n_N})$.
\subsection{Subclasses of games within quadratic games}\label{sec:quadGame_subClasses}
To both illustrate the breadth of quadratic games and provide exemplars of the game graph concept, we describe two important subclasses of games and their game graphs. 
\subsubsection{Finite horizon LQ game}~\label{ex:lqGame_setUp}
Given initial state $z^0 \in \reals^{m}$ and horizon $T$, each player $i$ in an $N$-player, finite-horizon LQ game selects an action sequence $(u_i^0,\ldots, u_i^{T-1})$ with $u_i^t \in \reals^{m_i}$  in order to minimize a cumulative state and control cost subjected to state dynamics:
\begin{equation}
 \begin{aligned}
   \underset{u_i^t\in \reals^{m_i}}{\min} & \ \tfrac{1}{2} \big( \textstyle\sum_{t = 0}^{T}(z^t)^\top Q_i z^t +  \textstyle\sum_{t = 0}^{T-1} (u_i^t)^\top R_i u_i^t \big)\\
   \text{s.t. }& \ z^{t+1} = Az^t + \textstyle \sum_{i = 1}^NB_iu_i^t, \ t = 0, \ldots, T-1. \label{eqn:LQGame_dynamics}
\end{aligned}   
\end{equation}
The LQ game defined by the collection of optimization problems~\eqref{eqn:LQGame_dynamics} for each $i\in [N]$ is equivalent to a one-shot quadratic game in which each player selects $U_i = [(u_i^0)^\top, \ldots, (u_i^{T-1})^\top]^\top\in \reals^{{n}_i}$ with $n_i=Tm_i$, in order to minimize their cost $f_i(U)$ defined by
\[\tfrac{1}{2}(\textstyle\sum_{j = 1}^N G_jU_j + Hz^0)^\top\bar{Q}_i(\sum_{j = 1}^N G_jU_j + Hz^0) + \tfrac{1}{2}U_i^\top\bar{R}_iU_i,\]
where $U = (U_1, \ldots, U_N)$ is the joint action profile, and the cost matrices are given by $\bar{Q}_i  = \diag\{Q_i, \ldots, Q_i\}$, 
\begin{equation}\label{eqn:LQ_game_parameters}
    \begin{aligned}
    G_i &= \begin{bmatrix}
        0 & \ldots & 0 \\
        B_i & \ldots & 0 \\
        \vdots & \ddots & \vdots \\
        A^{T-1}B_i & \hdots & B_i
        \end{bmatrix}\!, H = \begin{bmatrix}I \\\vdots\\A^{T} \end{bmatrix}\!, 
    \end{aligned}
\end{equation}
and $\bar{R}_i = \diag \{R_i, \ldots, R_i\}$.
This follows precisely from observing that the dynamics are equivalent to $Z = \sum_{i=1}^NG_iU_i + Hz^0$ where $Z = [(z^0)^\top, \ldots, (z^T)^\top]^\top$. From here, it is straight forward to rewrite the optimization problem in \eqref{eqn:LQGame_dynamics} as $\min_{U_i}f_i(U)$. The LQ game is a potential game if and only if $Q_i = Q_j$ and $R_i = R_j$ for all $i, j \in [N]$. 


\textbf{LQ Game Graph.}\label{ex:lqGame_Graph}
Suppose each player uses step size $\gamma_i$. Since,  $D_{i} f_i(U)$ is given by
\begin{equation}\label{eqn:lqGame_derivative}
     \textstyle (G_i^\top\bar{Q}_iG_i + \bar{R}_i)U_i + G_i^\top\bar{Q}_i (\sum_{j \neq i}G_jU_j+Hz^0),
\end{equation}
 the learning dynamics~\eqref{eqn:gameGraph_dynamics} are equivalent to
\begin{equation}\label{eqn:lqGame_update}
U^{k+1}  = W U^k - \Gamma[\bar{Q}_1G_1, \ldots, \bar{Q}_NG_N]^\top Hz^0,
\end{equation}
where $W = I_n - M$, with $M \in \reals^{n\times n}$ a blockwise matrix having entries $M_{ij} = \gamma_iG_i^\top\bar{Q}_iG_j$ if $i \neq j$ and $M_{ij} = \gamma_i(G_i^\top\bar{Q}_iG_i + \bar{R}_i)$ otherwise. 

\subsubsection{Bilinear games}~\label{ex:zeroSumGame}
Bilinear games are an important class of games. For instance, a number of game formulations in adversarial learning have a hidden bilinear structure~\cite{vlatakis2019poincare}. In evaluating and selecting hyper-parameter configurations  in so-called \emph{test suites}, pairwise comparisons between algorithms are formulated as bimatrix games~\cite{balduzzi2018re, balduzzi2020smooth}. 
%

Formally, a \emph{two player bilinear game}\footnote{The bilinear game formulation and corresponding game graph for different gradient-based learning rules  easily  extend to an $N$-player setting, however the results in Sec.~\ref{sec:disturbance} are presented for two player games.}, a subclass of continuous quadratic games, is defined by $f_1(x_1,x_2)=x_1^\top Ax_{2}$ and $f_2(x_1,x_2)=x_1^\top B^\top x_2$  where $A\in \mb{R}^{n_1\times n_2}$ and $B\in \mb{R}^{n_2\times n_1}$ and $x_i\in \mb{R}^{n_i}$. 
%
Common approaches to learning in games~\cite{vlatakis2019poincare,bailey2019finite},  simultaneous and alternating gradient descent both
correspond to a linear system.

\textbf{Game graph for simultaneous gradient play.}  Players update their strategies simultaneously by following the gradient of their own cost with respect to their choice variable: 
\begin{equation}\label{eqn:simGrad_update}
x^{k+1}_1=x_1^k  - \gamma_1 A x_{2}^k,\ x_2^{k+1}  =x_2^k-\gamma_2 Bx_1^k
\end{equation} 
The simultaneous gradient play game graph 
is given by
\begin{equation}\label{eqn:simGrad_gameGraph}
    \textstyle W_s=\begin{bmatrix}I & -\gamma_1A\\
     -\gamma_2B&I\end{bmatrix}.
 \end{equation} 

\textbf{Game graph for alternating gradient play.} 
In zero-sum bilinear games, it has been shown that alternating gradient play has better convergence properties~\cite{bailey2019finite}.
Alternating gradient play is defined by  
\begin{equation}\label{eqn:algGrad_update}
x_1^{k+1} =  x_1^k - \gamma_1Ax^k_2, \ x_2^{k+1} = x_{2}^{k}-\gamma_2Bx_{1}^{k+1}
\end{equation}
Examining the second player's update, we see that $x_2^{k+1}=(I+\gamma_1\gamma_2BA)x_2^k-\gamma_2 Bx_1^k$. The game graph in this case is defined by
 \begin{equation}\label{eqn:altGrad_gameGraph}
     W_a=\begin{bmatrix}I & -\gamma_1A\\
     -\gamma_2B&I+\gamma_1\gamma_2BA\end{bmatrix}.
 \end{equation} 
\begin{remark}Convergence  of~\eqref{eqn:simGrad_update} and boundedness of~\eqref{eqn:algGrad_update} depend on choosing appropriate step sizes $\gamma_{1}$ and $\gamma_2$~\cite{chasnov:2019aa,bailey2019finite}. We consider disturbance decoupling for settings such as these where the undisturbed dynamics are convergent. 
\end{remark}
\section{Disturbance Decoupling on Game Graph}\label{sec:disturbance}
In this section, we derive the necessary and sufficient condition that ensures decoupling of gradient disturbance from the learning trajectory of a subset of players. We emphasize that the condition holds for disturbances with arbitrary magnitudes and functions.
This is a useful result because it provides guarantees on both the equilibrium behavior and the learning trajectory under adversarial disturbance.


\begin{definition}[Complete disturbance decoupling] \label{def:completeDisturbance}
Given initial joint action $x^0 \in \reals^{n}$, game costs $(f_1, \ldots, f_N)$, step sizes $\Gamma \in \reals^{n\times n}$, suppose that player $i$'s gradient update is corrupted as in~\eqref{eqn:noiseModel}, then for player $j \neq i$, action $x_j$ is decoupled from the disturbance in player $i$'s gradient if the uncorrupted and corrupted dynamics, given respectively by
\begin{equation}\label{eqn:noiseNoNoiseTrajectory}
 x^{k+1}  = W x^k - \Gamma\bar{r}, \ y^{k+1} = W y^k - \Gamma \bar{r} - \Gamma d^k
\end{equation}
result in identical trajectories for player $j$ when $y^0 = x^0$. That is, $y^{k}_j = x^{k}_j$ holds for all $k \geq 0$, $d^k \in \mc{D}_i$, where
\[\textstyle\mc{D}_i  = \{d  = [d_1, \ldots, d_N ]^\top \in \reals^{\ndim} \ |\ d_j = 0, \forall \ j \neq i\}.\]
\end{definition}
\subsection{Algebraic condition}
We first derive an algebraic condition on the joint action space for disturbance decoupling.
Define $\mc{M}^\perp = \{x \in \reals^n \ | \ x^\top \tilde{x} = 0, \ \forall \  \tilde{x} \in \mc{M}\}$  and let $\range(A) = \{Ax \ |  \ x \in \reals^{n}\}$ denote the image of  $A\in \reals^{m\times n}$.
\begin{proposition}
\label{prop:disturbanceRejection}Consider an $N$-player quadratic game  $(f_1,\ldots, f_N)$ as in Definition \ref{def:quadGame} under learning dynamics as given by~\eqref{eqn:gradientUpdate}, 
where player $i$ experiences gradient disturbance as given by~\eqref{eqn:noiseModel}. Let  $\mc{S}(i) = \{ x = [x_1, \ldots, x_N]^\top \in \reals^{n}\  |\ x_j = 0,  \ \forall \ j \neq i \}$ be the joint action subset. For player $j \neq i$, the following statements are equivalent:
\begin{enumerate}
    \item[(i)] Player $j$ is disturbance decoupled from player $i$.
   \item[(ii)] 
  $ W^{k} v \in {\mc{S}}(j)^\perp$, $\forall\ v \in \mc{S}(i)$,  $\forall\ 0 \leq k < n$.
\item[(iii)] $\range(W^kE) \subseteq \range(Y)$, $\forall\ 0 \leq k < n$, where $E \in \reals^{n\times n_i}$ and $Y \in \reals^{n\times (n - n_j)}$ are matrices such that $\range(E) = \mc{S}(i)$ and $\range(Y) = \mc{S}(j)^\perp$.
\end{enumerate}
\end{proposition}
\begin{proof}
For a  quadratic game $(f_1,\ldots, f_N)$, the learning dynamics without and with disturbances reduce to the equations  in~\eqref{eqn:noiseNoNoiseTrajectory}.
Given initial joint action $x^0$, 
\begin{align*}
\textstyle x^k & = W^k x^0 - \textstyle \begin{bmatrix} W^{k-1} & \ldots & W^0\end{bmatrix}\Gamma\begin{bmatrix}\bar{r}^\top & \ldots, & \bar{r}^\top \end{bmatrix}^\top, \\ \textstyle y^k &=  
     x^k  - \begin{bmatrix} W^{k-1} & \ldots & W^0\end{bmatrix}\Gamma \begin{bmatrix}(d^0)^\top & \ldots, & (d^{k -1} )^\top \end{bmatrix}^\top\!.   %
\end{align*}
Then, Definition~\ref{def:completeDisturbance} is equivalent to $\textstyle\sum_{l = 0}^{M-1}W^{M-l-1} d^l \in {\mc{S}}(j)^{\perp}$ satisfied for $M \geq 1$ and $d^l \in \mc{S}(i)$. 
Since the condition holds for all $M\geq 1$, it is equivalent to $W^kd^l \in  {\mc{S}}(j)^{\perp}$ for all $k \geq 0$ and $d^l \in \mc{S}(i)$. 
This is then equivalent to $\textstyle W^kd^l \in  {\mc{S}}(j)^{\perp}$  for all $0 \leq k < n$ and $d^l \in \mc{S}(i)$.  
To see this equivalence, consider the following result from Cayley-Hamilton theorem,
%
$W^k = \sum_{l = 0}^{n-1}\alpha_l W^l$ for some $\alpha_l \in \reals$. Thus, for $k \geq n$ and any $d \in \mc{S}(i)$, $W^k d = \sum_{l = 0}^{n-1}W^l \alpha_l d = \sum_{l = 0}^{n-1}W^l \hat{d}_l$ where $\hat{d}_l = \alpha_l d  \in  {\mc{S}}(i)$ for $l=0,\ldots, n-1$, which implies that $W^kd \in  {\mc{S}}(j)^{\perp}$. This concludes the equivalence.

Finally, we note that $(iii)$ is a restatement of $(ii)$. Furthermore, $(iii)$ can be verified in polynomial time.
\end{proof}
\begin{remark}
In connection to geometric control theory, condition $(iii)$ of Proposition~\ref{prop:disturbanceRejection} is equivalent the fact that $\range( [E, \ldots, W^{n-1}E])$, the smallest $W$-invariant subspace containing $\range(E)$, must be a subset of $\mc{S}(j)^{\perp}$~\cite[Thm 4.6]{trentelman2012control}.
\end{remark}
\subsection{Graph-theoretic condition}
Next we derive the graph-theoretic condition on the joint action space for disturbance decoupling.
\begin{theorem}\label{thm:disturbanceDecoupling}Consider an $N$-player quadratic game $(f_1,\ldots, f_N)$ as in Definition \ref{def:quadGame} under learning dynamics as given by~\eqref{eqn:gradientUpdate}, where player $i$ experiences gradient disturbance as given by~\eqref{eqn:noiseModel}. Player $j \neq i$ is disturbance decoupled if and only if the path weights of paths with length $k$ satisfy
\begin{equation}\label{eqn:pathCondition}
\underset{ p \in \mc{P}^k_{ij}}{\sum}\overset{k-1}{\underset{l=0}{\prod}} W_{v_{l+1}, v_l} = 0,\ \forall \ 0 < k < n,
\end{equation}
where $(v_l,v_{l+1})$ denotes consecutive nodes on path $p = (i, v_1,\ldots, v_{k-1}, j)$. 
\end{theorem}
\begin{proof}
The result follows from equivalence between Proposition~\ref{prop:disturbanceRejection} condition $(ii)$ and \eqref{eqn:pathCondition}.
Note that $x \in \mc{S}(i)$ is equivalent to $x_\ell = 0$ for all $\ell \neq i$, and  $ W^k x \in \mc{S}(j)^\perp$ is equivalent to $(W^k x)_j = 0$ for all $n > k \geq 0$. We prove the result by induction. For $k = 0$, $(W^0x)_j = 0$ $\forall \ x \in \mc{S}(i)$ holds if and only if $i \neq j$. For $k > 0$, $(W^kx)_j = 0$ $\forall \ x \in \mc{S}(i)$ is equivalent to  $i \neq j$ and $(W^k)_{ji} = 0$.
Suppose that for $i, j \in [N]$, $(W^{k})_{ji}$ is the sum of path weights over all paths of length $k$, originating at $i$ and ending at $j$, then $(W^{k+1})_{ji}$ is the sum of path weights over all paths of length $k+1$, originating at $i$ and ending at $j$. Let $W^{k} = M$, then $(W^{k+1})_{ji} = \sum_{q \in [N]} M_{jq}W_{qi}$, where $M_{jq}W_{qi}$ is the sum of path weights over all paths of length $k+1$ from $i$ to $j$ each of which contains $v_{1} = q$. Since we sum over $q \in [N]$, we conclude that $(W^{k+1})_{ji}$ is the sum of all paths weights of length $k+1$ from $i$ to $j$, i.e., $(i, v_1,\ldots, v_{k}, j) \in \mc{P}_{ij}^{k+1}$.
\end{proof}
The concept of disturbance decoupling is quite counter-intuitive: any change in player $i$'s action does not affect player $j$'s action, despite $f_j$ being implicitly dependent on $x_i$ through the network of player cost functions. As we see from the proof of Theorem~\ref{thm:disturbanceDecoupling}, this situation arises when the dependencies `cancel' each other out, i.e. the sum of path weights from $i$ to $j$ is always zero for equally lengthed paths. 

\begin{example}[Disturbance decoupled players]\label{ex:DD_4Players}
Consider a $4$ player quadratic game where $x_i \in \reals$ and the game graph is given by Figure~\ref{fig:DD_example}. Edge weights $\alpha$, $\beta$, $\gamma$, and $\delta \in \reals$, while each self loop has weight $w_i > 0$. Paths of length $k \leq 4$ from player $1$ to player $4$ are enumerated as $\mc{P}_{14}^1 = \{\emptyset \}$, $\mc{P}_{14}^2 = \{(1,2,4), (1,3,4)\}$, and $\mc{P}_{14}^3 = \{ (1,1,2,4), (1,1,3,4), (1,2,2,4), (1,3,3,4)$ , $(1,2,4,4), (1,3,4,4) \}$. 
To satisfy Theorem~\ref{thm:disturbanceDecoupling}, the sum of path weights for each $\mc{P}_{14}^k$ must be $0$ for $0 < k < 4$.
There are no paths of length one, summation for $k= 2$ implies the criteria $\alpha\gamma + \beta\delta = 0$, and summation for $k = 3$ implies the criteria $(w_1 + w_2 + w_4)\alpha\gamma + (w_1 + w_3 + w_4)\beta\delta = 0$. If $w_2 = w_3$, $\alpha\gamma + \beta\delta = 0$ is necessary and sufficient for disturbance decoupling between player 1 and player 4. 
\end{example}
\begin{figure}
    \centering
    \includegraphics[width = 0.31\columnwidth]{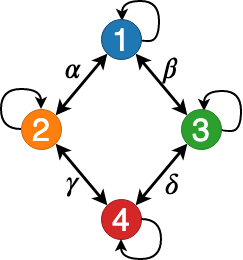}
    \caption{A simple game graph between four players}
    \label{fig:DD_example}
    \vspace{-10pt}
\end{figure}
\begin{remark}
Disturbance decoupling  is a structural property of the game  in terms of disturbance propagation and attenuation. An open research problem is linking this structural property to  robust decision making under   uncertainties in cost parameters $P_i$, $P_{ij}$ and step sizes $\gamma_i$.
\end{remark}
The following corollary specializes to the class of potential games~\cite{monderer1996potential}, which  arise in many  applications~\cite{paccagnan,lutati2014congestion,alpcan2010network}.
\begin{corollary}\label{cor:DD_zerosum_potential} 
Consider an $N$-player quadratic potential game under learning dynamics as given by~\eqref{eqn:gradientUpdate}, 
where player $i$ experiences gradient disturbance as given by~\eqref{eqn:noiseModel}. Player $i$ is disturbance decoupled from player $j \neq i$ if and only if player $j$ is also disturbance decoupled from player $i$. 
\end{corollary}
\begin{proof}
In a potential game graph, $W_{ij} = \frac{\gamma_i}{\gamma_j} W_{ji}^\top$.
Therefore, a path $p$ with path weight $W_{j,v_{k-1}}\ldots W_{v_{1},i}$ is equivalent to
\begin{align*}
    &\frac{\gamma_j}{\gamma_{v_{k-1}}} W_{v_{k-1}j}^\top \frac{\gamma_{v_{k-1}}}{\gamma_{v_{k-2}}} W_{v_{k-2}v_{k-1}}^\top \ldots \frac{\gamma_{v_{1}}}{\gamma_i} W_{i,v_{1}}^\top \\
    =& \frac{\gamma_j}{\gamma_{i}} W_{i,v_{1}}\ldots W_{v_{k-1}j},
\end{align*}
where $\frac{\gamma_j}{\gamma_{i}}$ scales all paths weights from $i$ to $j$. Since $\gamma_{j}, \gamma_i > 0$, $\frac{\gamma_j}{\gamma_{i}} > 0$. Therefore,~\eqref{eqn:pathCondition} holds from player $i$ to player $j$ if and only if it holds from player $j$ to player $i$.
\end{proof}
 
\begin{corollary}\label{cor:lqGame_DD}
Consider an $N$-player finite horizon LQ game as in~\eqref{eqn:LQGame_dynamics} under learning dynamics as given by~\eqref{eqn:lqGame_update}, where player $i$ experiences gradient disturbance as given by~\eqref{eqn:noiseModel}, if disturbance decoupling holds between player $j$ and  gradient disturbance from player $i$, then 
\begin{equation}\label{eqn:lqGame_DDCondition}
\textstyle\begin{bmatrix}
B_j^\top \\
\vdots \\
B_j^\top (A^\top)^{T-1}
\end{bmatrix}Q_j\begin{bmatrix}B_i & \cdots & A^{T-1}B_i\end{bmatrix} = 0.
\end{equation}
If $Q_j$ is positive definite  and $T\! \geq \! m$, the controllable subspace of $(\tilde{A},\tilde{B}_i)$ must lie in the unobservable subspace of $(\tilde{B}_j^\top, \tilde{A}^\top)$ where $\textstyle \tilde{A}=Q_j^{1/2}AQ_j^{-1/2} $, $\tilde{B}_i = Q_j^{1/2} B_i$, and $\tilde{B}_j = Q_j^{1/2} B_j$.
\end{corollary}
\begin{proof}
For player $j$ to be disturbance decoupled from player $i$, edge $(i,j)$ cannot exist, i.e. $-\gamma_jG_j^\top \bar{Q}_jG_i = 0$ from~\eqref{eqn:LQ_game_parameters}. Expanding $G_j^\top \bar{Q}_jG_i$ $= M \in \reals^{n_j \times n_i}$, $M_{pq} \in \reals^{m_j\times m_i}$ is given by 
${\sum^{T-1}_{t = \min\{p, q\}} } B_j^\top (A^\top)^{t- p} Q_j A^{t - q} B_i$.
We unwrap these conditions starting from $p = T-1$, $q = T-1$; in this case $M_{pq} = B_j^\top Q_j B_i = 0$ is necessary. Then we consider $M_{T-2, T -2}$ $= B_j^\top A^\top Q_j A B_i + B_j^\top Q_j B_i = 0$, which implies that $ B_j^\top A^\top Q_j A B_i $ is necessary. Subsequently, this implies that all $ B_j^\top (A^\top)^t Q_j A^t B_i = 0$ is necessary for $t \in [0, T)$. Similarly, we note that $M_{T-1, q} = B_j^\top Q_j A^{q}B_i = 0$ and $M_{p, T-1} = B_j^\top (A^\top)^{p} Q_j B_i = 0$. From these we can use the rest of $M$ to conclude that $B_j^\top (A^\top)^{p} Q_j A^{q} B_i = 0$ for any $p, q \in [0, T)$. This condition is equivalent to~\eqref{eqn:lqGame_DDCondition}.
\end{proof}
We apply Theorem~\ref{thm:disturbanceDecoupling} to two player bilinear games and prove a necessary condition for disturbance decoupling between different coordinates of each player's action space that is \emph{independent} of players' step sizes.
\begin{corollary}\label{cor:2pZeroSum}
Consider a two player bilinear game under learning dynamics~\eqref{eqn:simGrad_update} and~\eqref{eqn:algGrad_update}, where coordinates $x_{1,i}$ and $x_{2,i}$ experience gradient disturbance as given by~\eqref{eqn:noiseModel}. If $j \neq i$ and coordinate $x_{1,j}$ is disturbance decoupled from coordinate $x_{1,i}$, $(A,B)$ must satisfy $\sum_{\ell = 1}^{n_2} b_{\ell i} a_{j\ell}= 0$, 
where $a_{pq}$ and $b_{pq}$ denote the $(p, q)^{th}$ elements of $A$ and $B$, respectively. Similarly, if $j \neq i$ and coordinate $x_{2,j}$ is disturbance decoupled from coordinate $x_{2,i}$, $(A,B)$ must satisfy $\sum_{\ell = 1}^{n_1} b_{j \ell} a_{\ell i}= 0$.
\end{corollary}
\begin{proof}
We construct games played by $n_1 + n_2$ players with actions $\{x_{1,1}, \ldots, x_{1, n_1},x_{2,1}, \ldots, x_{2, n_2}\}$ and whose game graphs are identical to $W_s$~\eqref{eqn:simGrad_gameGraph} and $W_a$~\eqref{eqn:altGrad_gameGraph}. First consider disturbance decoupling of $x_{1,j}$ from $x_{1,i}$. In both learning dynamics, $\{x_{1,1}, \ldots, x_{1, n_1}\}$ do not have any edges between players. Therefore, paths between $x_{1,i}$ and $x_{1,j}$ with length $2$ is given by $\mc{P} = \{(x_{1,i}, x_{2,\ell}, x_{1,j}) \ | \ \ell \in [n_2]\}$. We sum path weights over $\mc{P}$ to obtain $ \sum_{\ell = 1}^{n_2} b_{\ell i} a_{j\ell}= 0$ for disturbance decoupling of $x_{1,j}$ from $x_{1,i}$  in~\eqref{eqn:simGrad_update} and~\eqref{eqn:algGrad_update}. A similar argument follows for disturbance decoupling of $x_{2,j}$ from $x_{2,i}$ in~\eqref{eqn:simGrad_update}. 
For disturbance decoupling of $x_{2,j}$ from $x_{2,i}$ in~\eqref{eqn:algGrad_update}, we note that a edge from $x_{2,i}$ to $x_{2,j}$ exists with weight $\gamma_1\gamma_2(BA)_{ji}$ when $j \neq i$. Disturbance decoupling requires $\gamma_1\gamma_2(BA)_{ji}=0$, therefore  $\sum_{\ell = 1}^{n_1} b_{j \ell} a_{\ell i} \!=\! 0$.
\end{proof}
\begin{corollary}Consider a two player bilinear game under learning dynamics~\eqref{eqn:simGrad_update} and~\eqref{eqn:algGrad_update}, where coordinates $x_{1,i}$ and $x_{2,i}$ experience gradient disturbance as given by~\eqref{eqn:noiseModel}. If coordinate $x_{2,j}$ is disturbance decoupled from coordinate $x_{1,i}$, $(A,B)$ must satisfy $b_{ji} = 0$ and $\sum^{n_2}_{q=1}b_{qi}\sum_{\ell=1}^{n_1}a_{\ell q}b_{j\ell} = 0$,
where $a_{pq}$ and $b_{pq}$ denote the $(p,q)^{th}$ elements of $A$ and $B$, respectively. If coordinate $x_{1,j}$ is disturbance decoupled from coordinate $x_{2,i}$, $(A,B)$  must satisfy $a_{ji} = 0$ and $\sum^{n_1}_{q=1}a_{qi}\sum_{\ell=1}^{n_2}b_{\ell q}a_{j \ell} \!=\! 0$.
\end{corollary}
\begin{proof}
We construct games played by $n_1 + n_2$ players with actions $\{x_{1,1}, \ldots, x_{1, n_1},x_{2,1}, \ldots, x_{2, n_2}\}$ and whose game graphs are identical to $W_s$~\eqref{eqn:simGrad_gameGraph} and $W_a$~\eqref{eqn:altGrad_gameGraph}.
In both learning dynamics, disturbance decoupling requires no direct path between the decoupled players. Therefore $a_{ji} = 0$ or $b_{ji} = 0$. 

Consider disturbance decoupling of $x_{1,j}$ from $x_{2,i}$ in~\eqref{eqn:simGrad_update}, paths of length $3$ from $x_{2,i}$ to $x_{1,j}$ without self loops is given by $\mc{P} = \{ (x_{2,i}, x_{1,q}, x_{2,\ell}, x_{1,j})\  | \ q \in [n_1]$, $\ell \in [n_2] \}$. A path of length $3$ with self loops must also include $(x_{2,i}, x_{1,j})$, whose weight is $0$. We sum path weights over $ p \in \mc{P}$ to obtain $\sum^{n_1}_{q=1}a_{qi}\sum_{\ell=1}^{n_2}b_{\ell q}a_{j \ell} = 0$. A similar argument is made for disturbance decoupling of $x_{2,j}$ from $x_{1,i}$ in~\eqref{eqn:simGrad_update}. 

Consider disturbance decoupling of  $x_{2,j}$ from $x_{1,i}$ in~\eqref{eqn:algGrad_update}, paths of length $2$ from $x_{1,i}$ to $x_{2,j}$ without self loops is given by $\mc{Q} = \{(x_{1,i}, x_{2,q}, x_{2,j}) \ | \ q \in [n_2]\}$. A path of length $2$ with self loops must also include $(x_{1,i}, x_{2,j})$, whose weight is $0$. Weight of $(x_{2,q}, x_{2,j})$ is given by $\gamma_1\gamma_2(BA)_{jq}$ $= \gamma_1\gamma_2\sum_{\ell=1}^{n_1} b_{j\ell}a_{\ell q}$. We sum path weights over $p \in \mc{Q}$ to obtain $\textstyle\sum^{n_2}_{q=1}b_{qi}\sum_{\ell=1}^{n_1}a_{\ell q}b_{j\ell} = 0$. A similar argument is made for disturbance decoupling of $x_{1,j}$ from $x_{2,i}$ in~\eqref{eqn:algGrad_update}.
\end{proof}

\section{Numerical Example}\label{sec:sim}
We provide an example of disturbance decoupling in a LQ game. 
Consider a tug-of-war game in which a single target $z \in \reals^2$ is controlled by four players. We assume that player $i$ can move $z$ along vector $B_i \in \reals^2$ by $u_i \in \reals$, and that $z$ is stationary without any player input, i.e., $A = I$. Starting with a randomized initial condition $z^0$, at each step $t$, the target moves according to the dynamics $z^{t+1} = z^t + \sum_{i =1}^4B_iu^t_i$ where
 $B_1 = [1,0]^\top$, $B_2 = \textstyle[\frac{1}{\sqrt{2}}, \frac{1}{\sqrt{2}}]^\top$, $B_3 = \textstyle[\frac{-1}{\sqrt{2}}, \frac{1}{\sqrt{2}}]^\top$, $B_4 = \textstyle[0, 1]^\top$. 
Each
player $i$'s cost function is given by
\begin{equation*}\label{eqn:tugCost}
\textstyle\tfrac{1}{2}\norm{z^{9} - c_i}_2^2 + \sum_{t = 0}^{8}\tfrac{1}{2}\norm{z^t - c_i}_2^2 + 10\norm{u^t_i}_2^2
\end{equation*}
which describes player $i$'s objective to move target $z$ towards
$c_i \in \reals^2$ in a finite time $T = 10$ by using minimal amount of control. 
By designing the game dynamics to satisfy Theorem~\ref{thm:disturbanceDecoupling}, we ensure that player $4$'s action is disturbance decoupled from player $1$'s.  

Using the equivalent formulation as described in Section~\ref{ex:lqGame_setUp}, 
$D_i f_i(U)$ $\textstyle = (G_i^\top\bar{Q}_iG_i + \bar{R}_i)U_i  + \sum_{j \neq i}G_i^\top\bar{Q}_i (G_jU_j  +Hz^0 - C_i) $ where $C_i = [c_i^\top, \ldots, c_i^\top]^\top$.
Hence, the learning dynamics are  $U^{k+1}  = W U^{k} + \Gamma\bar{Q}_i[G_1, \ldots, G_N]^\top[(Hz^0 - C_1)^\top, \ldots, (Hz^0 - C_N)^\top]^\top$, where $W_{ij} = G_i^\top\bar{Q}_iG_j = E \otimes B_i^\top B_j$
with $B_1^\top B_2= B_1^\top B_3 $ $= B_2^\top B_4 = \frac{1}{\sqrt{2}}$,  $B_2^\top B_3 = B_1^\top B_4 = 0$, $B_3^\top B_4 = -\frac{1}{\sqrt{2}}$, $B_1^\top B_1 = B_2^\top B_2 = B_3^\top B_3 = B_4^\top B_4 = 1$, 
and \[\textstyle E = \begin{bmatrix}9 & 8 & 7 & \ldots &1\\
8& 8& 7  &\ldots & 1\\
7 & 7 & 7 &\ddots & 1\\
\vdots & & \ddots& & 1\\
1 & & \hdots & & 1\end{bmatrix} \in \reals^{9\times 9}.\]

To ensure convergence of the undisturbed learning dynamics~\cite{chasnov:2019aa}, we use uniform step sizes such that $\Gamma=\text{blkdiag}(\gamma_1 I, \ldots, \gamma_4I)$ with $\gamma_i = \frac{\sqrt{\alpha}}{\beta}$, where $\alpha = \lambda_{\min}\big(\frac{1}{4}(W + W^\top)^\top(W + W^\top)\big)$ and $\beta = \lambda_{\max}\big(W^\top W\big)$ with $\lambda_{\max}(\cdot)$ and $\lambda_{\min}(\cdot)$ denoting the   maximum and minimum eigenvalues of their arguments,  respectively. 
The associated game graph is given in Figure~\ref{fig:DD_example}, where $\alpha = \beta = \gamma  =\frac{1}{\sqrt{2}} E$ and $\delta = -\frac{1}{\sqrt{2}} E$.
A path $p =(1, v_1, \ldots, v_{k-1}, 4)$ of length $k$ must have path weight $(\frac{-1}{\sqrt{2}})^{m_\delta}(\frac{1}{\sqrt{2}})^{m_\gamma}E^k$, where $m_\delta$ ($m_\gamma$) denotes the number of times the edge with weight $\delta$ ($\gamma$) is traversed in $p$.

Disturbance decoupling between players $1$ and $4$ is guaranteed if all paths of length $k \in (0,36)$ satisfy~\eqref{eqn:pathCondition}. We can numerically verify that Proposition~\ref{prop:disturbanceRejection} is satisfied or make the following graph-theoretic observations based on Theorem~\ref{thm:disturbanceDecoupling}. First, due to the symmetry within the game graph, the existence of path $p=(1, v_1, \ldots, v_{k-1}, 4)$ with path weight $ \textstyle L = (\frac{-1}{\sqrt{2}})^{m_\delta}(\frac{1}{\sqrt{2}})^{m_\gamma}E^k$ implies the existence of path $\hat{p}=(1, \hat{v}_1, \ldots, \hat{v}_{k-1}, 4)$ with path weight $\hat{L} = (\tfrac{-1}{\sqrt{2}})^{
\hat{m}_\delta}(\tfrac{1}{\sqrt{2}})^{\hat{m}_\gamma}E^k$, where $m_\gamma = \hat{m}_\delta$ and $m_\delta = \hat{m}_\gamma$.  Second, since edges $(3,4)$ and $(2,4)$ form a cut between player $1$ and player $4$ in the game graph, any path between them has the property that $m_\gamma + m_\delta $ is odd. From these observations, we can conclude that $L = - \hat{L}$. Since each path $p$ of length $k$ and weight $L$ can be paired with path $\hat{p}$ of equivalent length $k$ and weight $\hat{L} = -L$, we conclude that all path sets $\mc{P}_{14}^k$ where $k > 0$ must satisfy Theorem~\ref{thm:disturbanceDecoupling}.

To numerically verify disturbance decoupling, we simulate the uncorrupted learning trajectory of $z$, shown in the left plot of Figure~\ref{fig:SimRes} in purple. We then inject a random disturbance into player $1$'s gradient updates as given by~\eqref{eqn:noiseModel} with increasing magnitude, and observe its effects on each player's action. A sample corrupted trajectory is shown in the left plot of Figure~\ref{fig:SimRes} in brown.
In the bottom right plot of Figure~\ref{fig:SimRes}, we show the total error in each player's action from to the uncorrupted optimal action. We observe that player $4$ does not deviate from the optimal action, while player $1$'s action error increases as the disturbance magnitude increases. 
We note that these results hold despite the fact that gradient-based learning no longer converges. In the top right plot of Figure~\ref{fig:SimRes}, individual player costs are compared in one round of gradient-based learning where $\norm{d_i} \leq 50$ is injected. Interestingly, despite action remaining uncorrupted, player $4$'s cost \emph{is} disturbance affected. Note that the disturbance decoupling in actions does not necessarily  imply disturbance  decoupling in  costs.

\begin{figure}
  \centering
  \begin{subfigure}[b]{0.44\columnwidth}
    \includegraphics[width=\columnwidth]{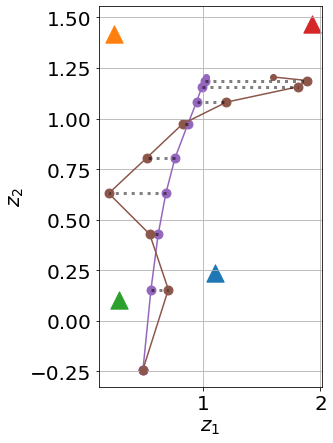}
  \end{subfigure}
  \begin{subfigure}[b]{0.4\columnwidth}
    \includegraphics[width=\columnwidth]{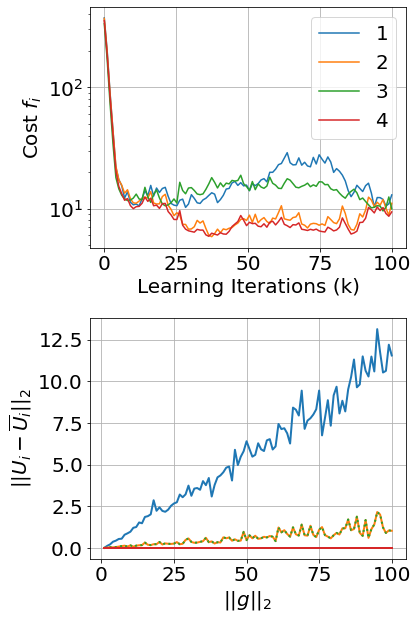}
  \end{subfigure}
  \caption{
  Left: Trajectory of $z$ with and without disturbances. Players' preferred destinations are given by triangles. Top right: Players' game costs during learning. Bottom right: Players' control error as a function of disturbance magnitude.}
  \label{fig:SimRes}
\end{figure}

\section{Conclusion}\label{sec:conclusion}
In this paper, we investigated and characterized the effects of gradient disturbances on an $N$--player gradient-based learning dynamics. For quadratic games, we defined disturbance decoupling for arbitrary disturbances, and showed the cost coupling structure is crucial in facilitating decoupling individual player's action from input disturbance. Our future work aims to leverage these analysis results to design incentives for players to ensure disturbance decoupling. 
\bibliographystyle{IEEEtran}
\bibliography{reference}

\end{document}